%% file: access.tex
\documentclass{ieeeaccess}
\usepackage{cite}
\usepackage{algorithmic}
\usepackage{graphicx}
\usepackage{textcomp}
\usepackage{cite}
\usepackage{amsmath,amssymb,amsfonts,amsthm}
\usepackage{algorithmic}
\usepackage{graphicx,color}
\usepackage[caption=false]{subfig}

\newcommand{\px}{\mathsf{x}}
\newcommand{\py}{\mathsf{y}}

\input{testmathcommand.tex}

\newcommand{\R}{\mathbb{R}}

\usepackage{url}
\usepackage{hyperref}
\newcommand{\rEarth}{R_{\oplus}}

\usepackage{bm}
\makeatletter
\AtBeginDocument{\DeclareMathVersion{bold}
\SetSymbolFont{operators}{bold}{T1}{times}{b}{n}
\SetSymbolFont{NewLetters}{bold}{T1}{times}{b}{it}
\SetMathAlphabet{\mathrm}{bold}{T1}{times}{b}{n}
\SetMathAlphabet{\mathit}{bold}{T1}{times}{b}{it}
\SetMathAlphabet{\mathbf}{bold}{T1}{times}{b}{n}
\SetMathAlphabet{\mathtt}{bold}{OT1}{pcr}{b}{n}
\SetSymbolFont{symbols}{bold}{OMS}{cmsy}{b}{n}
\renewcommand\boldmath{\@nomath\boldmath\mathversion{bold}}}
\makeatother

\def\BibTeX{{\rm B\kern-.05em{\sc i\kern-.025em b}\kern-.08em
    T\kern-.1667em\lower.7ex\hbox{E}\kern-.125emX}}

\begin{document}
\doi{}

\title{Spatial and Temporal Correlation of Interference in a Narrow Multibeam LEO Satellite Random Access Network}
\author{\uppercase{Ilari Angervuori}\authorrefmark{1}, \IEEEmembership{Student Member, IEEE},
  \uppercase{Abid Afridi}\authorrefmark{1}, \IEEEmembership{Student Member, IEEE}, AND \uppercase{Risto Wichman},
  \authorrefmark{1},
  \IEEEmembership{Senior Member, IEEE}}

\address[1]{School of Electrical Engineering, Aalto University, Espoo, 02150, Finland \{ilari.angervuori;abid.afridi;risto.wichman\}@aalto.fi}

\markboth
{Ilari Angervuori \headeretal}
{Ilari Angervuori \headeretal}

\corresp{Corresponding author: Ilari Angervuori (e-mail: ilari.angervuori@aalto.fi).}

\begin{abstract}
  Interference is a limiting factor in the emerging dense low Earth orbit (LEO) networks. In the LEO network, the interference is spatially and temporally correlated. At narrow-beam LEO base stations (BSs), spatial interference can vary significantly, and multipath fading introduces temporal variation. While developing novel stochastic geometry analysis in a multibeam scenario, we explore spatio-temporal interference correlation in the LEO uplink. We derive a closed-form expression for the spatio-temporal interference correlation coefficient. As an application of the analysis, we show that the signal-to-interference ratio (SIR) entails significant spatial clustering. In this regard, we demonstrate that an appropriately designed grant-free random access scheme, particularly slotted ALOHA, can mitigate spatial SIR clustering over the beams while preserving average throughput. Furthermore, we propose a novel gamma distribution model for the interference power distribution and a Lomax distribution model for the SIR.
\end{abstract}


\titlepgskip=-21pt

\maketitle

\section{INTRODUCTION}

The rapid proliferation of low Earth orbit (LEO) satellite constellations has transformed global connectivity. This has enabled low-latency communications for diverse applications, particularly in narrowband communications. However, dense LEO networks introduce significant interference challenges, especially in the uplink. Interference varies widely due to the mobility of LEO satellites, leading to highly dynamic channel conditions \cite{talgat2024stochastic}. Furthermore, beamforming and multibeam architectures are likely technologies in near-future LEO networks \cite{s21144877}. With a narrow-beamed LEO BS antenna, the channel variation is particularly significant \cite{10909705}.

Using stochastic geometry, multiple attempts have been made to characterize interference at a LEO base station (BS) from randomly distributed terrestrial terminals in the LEO uplink \cite{10909705,9347980, 9684552, modelinguplink, 9918046, 9511625, globecom2022,11161458}. The challenge is the spherical Earth geometry, which, without appropriate simplifications, yields undesirable, complicated expressions that lack insight. In particular, the expression for the Laplace transform of the interference RV, which is usually needed when evaluating other performance metrics such as signal-to-interference ratio (SIR) or throughput, becomes impractically complicated, containing unevaluated integrals and/or infinite sums, and hence deriving closed-form or analytical expressions even for the first two moments is either complex or strictly impossible.~\footnote{We use the terms ``closed-form'' and ``analytical'' in their strict sense. Roughly, closed-form expressions do not contain unevaluated integrals, infinite sums, or special functions. Analytic form can contain these expressions. However, the unevaluated integrals must be expressible with analytic special functions.} To address this issue, we propose a tractable interference modeling framework that provides \textit{explicit} (i.e., closed-form or analytical) solutions for the first two moments of the interference, provided that the signal attenuation distribution (such as fast-fading) admits such explicit expressions. We also formulate a probability generating functional (PGFL) (from which the Laplace transform directly follows), whose expression complexity depends only on the complexity of the cumulative distribution function (CDF) of the random signal attenuation of a typical terrestrial terminal. This provides more insight and lays a tractable framework for further performance metric evaluations, such as SIR or average link throughput. We propose an approximate gamma distribution interference model and a Lomax SIR model, which we show to be applicable to Rician fading.

To address the challenges posed by control signaling overhead, grant-free random access schemes have been proposed for LEO uplink, offering simplicity and scalability while aiming to optimize link performance \cite{Mandawaria22}. Grant-free access is a low-complexity scheme because it does not require scheduling requests and preliminary resource allocation, hence minimizing signaling overhead. Of these schemes, ALOHA is the simplest, which we will model and analyze in this work; however, it also has the most frequent packet collisions. The ALOHA network model provides pessimistic predictions of throughput compared to more sophisticated schemes.


 Stochastic geometry has emerged as a powerful tool for analyzing spatio-temporal interference correlations, providing tractable models for coverage probability, rate, and interference distribution. By modeling the transmitters as random point processes, such as the homogeneous Poisson point process (PPP), one can derive mathematical expressions for performance metrics in terrestrial networks that capture both long-term spatial dependencies and short-term temporal fluctuations induced by fading \cite{lu2021stochastic,5282357}. Building on these works, we lay a tractable framework for spatial and temporal interference analysis, including a closed-form expression for the Pearson correlation coefficient. We study the average throughput in a LEO random access uplink. We show that the grant-free random access, particularly slotted ALOHA, can preserve optimal average throughput while ensuring consistent link performance, leading to more resilient LEO networks.

To the best of our knowledge, this is the first theoretical characterization of spatio-temporal interference correlation functions and the study of grant-free random multi-access in LEO networks within a stochastic geometry framework. \textcolor{black}{The correlation functions leverage an additional spatial and temporal structure in the network models, which is crucial for comprehensive analysis and optimization of the medium access control (MAC) schemes.}

\subsection{Organization of the paper}

\textcolor{black}{In Section \ref{sec:systemmodel}, we formulate and compare the planar and spherical system models and introduce the fading and antenna pattern models. Section \ref{sec:analysis} formulates the analysis and provides expressions for the mean and variance of the interference. In Section \ref{sec:spatiotemporal}, we formulate a closed-form expression for the spatio-temporal Pearson correlation coefficient of the interference. In Section \ref{sec:tempcorlinkout}, we study the SIR distribution and its temporal correlation properties. The key insights of the results are summarized in Section \ref{sec:conclusions}.}

\section{SYSTEM MODELS}
\label{sec:systemmodel}

\textcolor{black}{We present two system models in the subsections:
\begin{enumerate}
\item[\ref{sec:planarsystemmodel}] A simplified planar system model for the mathematical analysis.
\item[\ref{sec:sphericalsystemmodel}] A more realistic spherical system model used in Monte Carlo simulations.
\end{enumerate}
The planar system model analysis approximates the spherical model.}


\subsection{PLANAR SYSTEM MODEL}
\label{sec:planarsystemmodel}
\begin{figure}[ht]%
  \centering
  \subfloat[\textcolor{black}{A LEO BS and the beam footprints, the typical beam pointed at the origin,} and the UEs in the planar model. ]{{\includegraphics[width=\linewidth]{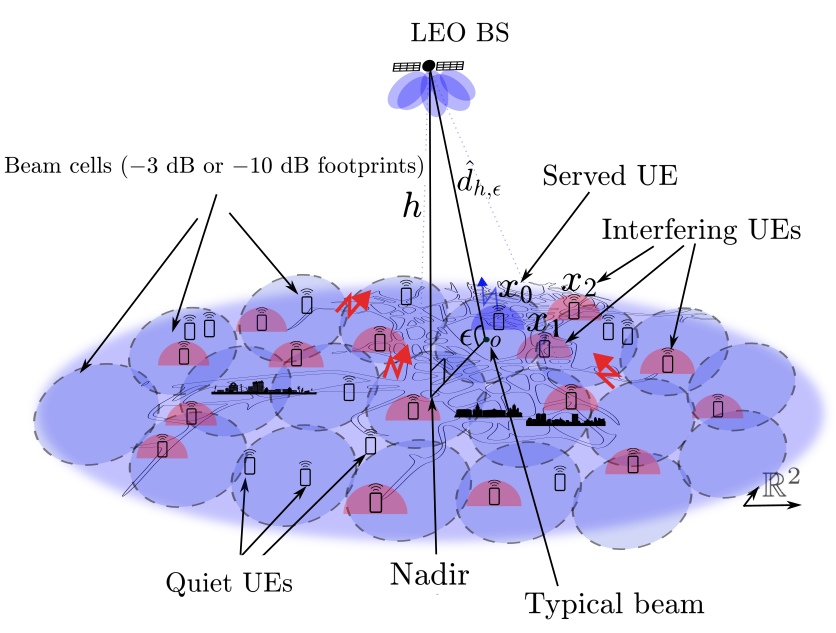}}
    \label{fig:UEsontheplane}
  } 
  \qquad
  \subfloat[\textcolor{black}{A LEO BS} as seen from the side. The transmitters are projected into the line $(0, \infty)$ according to their norm.]{{\includegraphics[width=\linewidth]{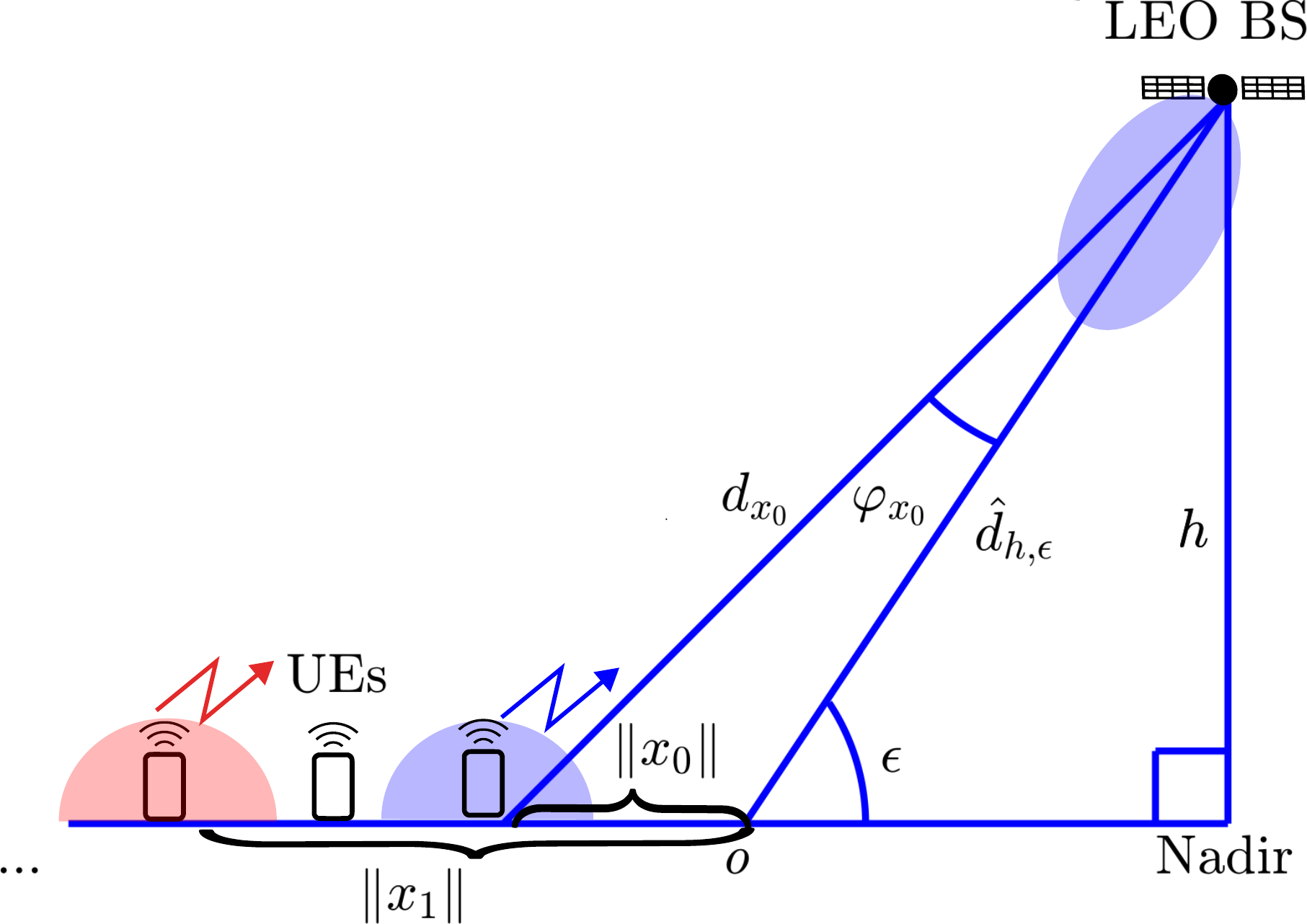}}     \label{fig:systemmodel}} 
  \caption{The simplified narrow-beam LEO uplink system model. The \textcolor{black}{typical LEO BS beam} antenna boresight is oriented towards \textit{o}, the focus point of the elliptical footprint. The omnidirectionally transmitting UEs $\{x_i\}$ are located according to a homogeneous PPP on the plane. The nearest transmitter, $x_0$, is the served UE.}
\end{figure}

A LEO narrowband multibeam uplink is considered with narrow Gaussian-beamed LEO base stations (BSs) in a \textcolor{black}{uniform single-altitude constellation from the perspective of a LEO BS serving the nearest user equipments (UE) to the beam center in each beam. The beams are fixed, and no beamforming techniques are used to track the served UEs while \textcolor{black}{the LEO BSs move} at orbital speed; this models a beam-hopping scenario.} In the analysis, we use a planar Earth model with the origin at $\textit{o} \triangleq (0,0) \in \R^2$, as depicted in Figs. \ref{fig:UEsontheplane} and \ref{fig:systemmodel}. \textcolor{black}{In the sketches, $\varphi_{x_0}$ is the angle between the nearest UE $\px_0$ and the antenna boresight directed at the origin, which represents \textcolor{black}{a LEO BS beam center}. $h$ is the altitude of \textcolor{black}{the LEO BSs}, and $\|\px_1\|$ and $\|\px_0\|$ are the Euclidean distances from the origin. These parameters are highly exaggerated in the figures. The elevation angle is $\epsilon$, and the distance of the LEO BS from the origin is $\hat{d}_{h,\epsilon}$; these parameters correspond to their counterparts in the corresponding spherical Earth model (see Appendix \ref{app:truegeom}).} The planar model is accurate for narrow beams because the effective beam footprint radius is much smaller than the Earth radius; hence, the Earth surface can be locally approximated as tangent at the sub-satellite point.


The LEO BS beams form a uniform layout, whereas the UEs form a homogeneous PPP $\Phi$ of density $\lambda$ on the plane $\R^2$ and have omnidirectional antennas. The beam centers are at most at the elevation angle $\epsilon = 35^{\circ}$ w.r.t the satellite. The block lengths are on the time scale of fading coherence (i.e., block fading). A \textit{use period} consists of multiple such blocks, but is brief enough so that \textcolor{black}{the LEO BSs move insignificantly} during the period. We consider a slotted ALOHA random access network in which each UE independently transmits with probability $p$ in each block, and otherwise is quiet. Because we assume narrow Gaussian beams that decay rapidly, the relevant UEs are spread over a relatively small area. Hence, the distance between a LEO BS and each transmitter is approximated as equal; $d_{\px}  \approx \hat{d}_{h,\epsilon}$ for all $\px \in \Phi$. Since the spatial path losses cancel in the SIR and the correlation coefficient, without loss of generality, $\hat{d}_{h,\epsilon}$ is normalized for each beam. Similarly, the received signal powers are normalized (farther beams increase their gain accordingly). Because of the narrow beams, all beam footprints are practically circular. To control interference, the other beams are adjusted so that footprints (or cells) are of equal size. \textcolor{black}{Furthermore, we will show in Corollary \ref{thm:densityoftheGP} that the theoretical interference statistics depend only on the cell size (equivalently, on the average number of UEs inside a beam footprint). Hence, if we keep the beam footprint sizes equal, it is warranted to talk about the ``typical beam'', which is directed at the origin, representing the statistics of any beam.}
\begin{rem}
  Unless otherwise stated, we refer to the term ``a cell'' as the $-3$ dB footprint of the beam. However, the cells do not have definite boundaries in the strict sense, and the served UE can be in the $-3$ dB footprint or, for example, in the $-10$ dB cell (i.e., the  main lobe footprint) outside the $-3$ dB cell, or even outside the main lobe. 
  \end{rem}


         \begin{table}[ht]
           \begin{center}
             \caption{Glossary of principal symbols} \label{tab:glossary}
             \begin{tabular}{ |c | p{5.9cm}|  }
               \hline
               \textbf{Symbol}& \textbf{Explanation}
               \\
               \hline
               $\Phi \subset \R^2$ & Homogeneous PPP on the plane. \\
               $\lambda$ &  Density of $\Phi$. \\
               $p \in (0,1]$ & Transmission probability. \\
                 $H_{\px}=H_{\px}(k)$ &  Fading gain of mean $1$ of a transmitter $\px \in \Phi$.\\
                 $G(\cdot)$ & LEO BS antenna gain.\\
                 $\mathcal{G},\mathcal{G}_1$ &Gain process (GP); latter denotation if $H_{\px}\equiv 1$. \\
                 $\Phi(k) \subset \Phi$ & Homogeneous PPP of the user equipments (UEs) scheduled at time $k$. \\
                 $\mathcal{G}(k) \subset \mathcal{G} $ &The gain process of the users $\Phi(k)$. \\
                 $\px$ & According to the context, a point in $\Phi$, $\mathcal{G}_1$ or $\mathcal{G}$.\\
                 $\|\px\|$ & The distance of $\px \in \Phi$ from the origin $\textit{o} =(0,0) \in \R^2$.  \\
                 $\px_0=\px_0(k)$ & The nearest UE to the \textcolor{black}{typical beam center} in $\Phi(k)$.  \\
                 $x,u,v$ & A point location in $\R^2$.  \\
                 $h$ & Altitude of the LEO BSs. \\
                 $\epsilon$ & Elevation angle of the LEO BSs w.r.t. the beam center. \\
                 $\varphi_{\text{RX}}=\varphi_{\text{RX}}(\epsilon)$ & Half-Width of the LEO BS $-3$ dB beam. \\
                 $\overline{\varphi^{}}_{\text{RX}}$ & Half-Width of the nadir-pointing LEO BS $-3$ dB beam. \\                 
                 $D_{h,\epsilon}$ & Scaling constant of $\|x\|$; $D_{h,\epsilon}= \sin^2(\epsilon)/h$.\\
                 $\kappa$ & A parameter that reflects the approximate mean number of UEs inside a LEO BS $-3$ dB beam footprint. \\
                 ${\tilde{\kappa}}$ & GP parameter; $\kappa/\log(2)$.\\
                 $\theta$ & SIR threshold for a successful transmission.\\
                 $I_{\textit{o}}=I_{\textit{o}}(k)$ & Total interference at the typical beam at time $k$.\\
                 $W$  & sidelobe interference constant. \\
                 \hline
             \end{tabular}             
           \end{center}
         \end{table}

         \textcolor{black}{Let $\varphi^{}_{\text{RX}}$ be the $-3$ dB beam half-width. For convenience, we will often refer to the ``beam half-width'' simply as ``beam width'' in the paper. We set $\varphi^{}_{\text{RX}}= \overline{\varphi}^{}_{\text{RX}} \triangleq 0.028$ $(= 1.6^{\circ})$ for the nadir-pointing beam, i.e., the beam that points to the sub-satellite point with the elevation $\epsilon =90^{\circ}$. The nadir-pointing beam width follows the ITU-R reference radiation pattern recommended by the International Telecommunication Union (ITU) \cite[ITU-R LEO reference radiation patterns]{ITURS1528}, of which the main lobe is indeed near-Gaussian \cite{10909705}[Fig. 2].}

         The Gaussian antenna pattern is at the kernel of the analysis approximates a Gaussian antenna pattern, is defined at the Euclidean distance $\|x\|$ as
         \begin{equation}
           \label{eq:antennagain}
           G(\|x\|)\overset{}{\triangleq} 2^{-(D_{h,\epsilon}\|x\|)^2/\varphi^2_{\textup{RX}}} \overset{(a)}{\approx}  2^{-\varphi_x^2/\varphi^2_{\textup{RX}}},
         \end{equation}
         where corresponds $\varphi_x$ to the angle between the antenna boresight and the location $x \in \R^2$. The scaling constant \begin{equation}D_{h,\epsilon} \triangleq \sin^2(\epsilon)/h \approx \varphi_{x}/\|x\|\end{equation} is the first-order coefficient of the Taylor expansion of $\varphi_{x}$, leveraging the approximation \eqref{eq:antennagain} (a). A detailed description of this approximation is provided in Appendix \ref{app:constantD}.

         Let $\Phi \subset \R^2$ be a homogeneous PPP representing the UEs on the Earth surface. Let $\{H_{\px}\}_{\px \in \Phi}$ be i.i.d. RVs with finite variance, representing power fading in this work. The gain process (GP) is defined as the Gaussian projection process
  \begin{equation}
    \label{eq:GP}
    \mathcal{G} \triangleq  \left\{H_{\px} G(\|\px\|):\px \in \Phi \right\}.
  \end{equation}
 Furthermore, denote $\mathcal{G}_1\triangleq\{G(\|x\|):\px \in \Phi \}$.

 Let $\Phi(k) \subset \Phi \subset \R^2$ be the set of transmitting UEs at time $k$. The corresponding GP is denoted as $\mathcal{G}(k) \subset \mathcal{G}$. 

 The total (``total'' because none of the UE signals are canceled) interference at time instant $k$ at the beam location $z \in \R^2$ is
\begin{align}
  I_z=I_z(k) &\triangleq \sum_{\px \in \Phi(k)}H_{\px}(k) G(\|\px-z\|) + W \\
  &= \sum_{\px \in \Phi}\mathbf{1}(\px \in \Phi(k))H_{\px}(k) G(\|\px-z\|) + W \nonumber,
\end{align}
where $\mathbf{1}(\cdot)$ is the indicator function modeling transmission with probability $p$, \textcolor{black}{and $W\geq 0$ is a constant power, representing the sidelobe interference component in this paper, which is determined empirically}. Each $H_{\px}(k)$ is assumed to be ergodic, and because of the i.i.d. property, we can refer to a typical fading gain (or, more generally, a typical signal) $H$. The homogeneous PPP is translation invariant, and the interference at each point is identically distributed; hence, \textcolor{black}{the interference at $\textit{o}$ provides all necessary statistics for the typical beam:}
\begin{align}
  \label{eq:totalinterferenceat0}
  I_{\textit{o}}=  I_{\textit{o}}(k) &= \sum_{\px \in \Phi(k)}H_{\px}(k) G(\|\px\|) +W \nonumber \\
  &= \sum_{\px \in \mathcal{G}_1(k)}H_{\px}(k) \px + W= \sum_{\px \in \mathcal{G}(k)}\px + W.
\end{align}
Since the typical LEO BS beam is directed at $\textit{o}$, the nearest (and served) UE is formally defined as
\begin{equation}
  \label{eq:x0def}
  \px_0=\px_0(k) \triangleq \arg \min \{ \px \in \Phi(k): \|\px\|\}.
\end{equation}
Further, denote the corresponding point in $\mathcal{G}_1$ by $\px_0$ (i.e., the largest gain excluding the fading gain).

\subsection{SPHERICAL System Model}
\label{sec:sphericalsystemmodel}
\textcolor{black}{
The Gaussian antenna pattern is defined as}
\begin{equation}
  G_{\textup{Gauss}}(\varphi)\overset{}{\triangleq} 2^{-\varphi^2/\varphi^2_{\textup{RX}}},
\end{equation}
\textcolor{black}{where $\varphi$ is the angle to the antenna boresight. }

To encompass the effect of sidelobes, the ITU-R \cite[LEO reference radiation pattern]{ITURS1528} antenna gain $G_\text{ITU-R}(\cdot)$ is also used in the simulations. It is given as
\begin{align}
  \label{eq:ITUpattern}
  &G_\text{ITU-R}[\varphi] \nonumber \\
  &\triangleq \begin{cases}
    10^{-3(\varphi^2/\varphi_{\text{RX}}^2)/10},0 \leq \varphi \leq 1.5 \cdot \varphi_{\text{RX}}^{\circ}  \\
    10^{-6.75/10 - 25\log_{10}\left(\frac{\varphi}{1.5 \cdot \varphi_{\text{RX}} } \right)/10},1.5 \cdot \varphi_{\text{RX}}   < \varphi \leq 20.4^{\circ} \\
    10^{-3},  \varphi> 20.4^{\circ}.
  \end{cases}
\end{align}
We consider that all UEs are within $\varphi <22^{\circ}$.

\textcolor{black}{The Gaussian antenna pattern approximates the main lobe ($-10$ dB) of the ITU-R antenna pattern (see, \cite{10909705}[Fig. 2].)}

\textcolor{black}{Explicitly, the total interference is defined in the spherical for the Gaussian beam and the ITU-R beams}
\begin{equation}
  \label{eq:mathringptot}
        {I}_{\textit{o}_E} =  \sum_{\px \in \Theta}  \frac{H_\px G_{\{\textup{Gauss, ITU-R}\}}(\varphi_\px)}{(d_{\px}/d_0)^{\gamma}} ,       
\end{equation}
\textcolor{black}{respectively. Here, $\Theta \subset E$ is a PPP on the sphere of radius $\rEarth = 6378$ km (generated as in \cite{10909705}[Sec. II-B]), $E$ representing the Earth surface, $\textit{o}_E$ is the beam center, $d_\px$ is the path loss of the UE $\px \in \Theta$, $\gamma = 4$ is the power path-loss exponent, and $d_0$ is a normalization parameter.}

\textcolor{black}{In the spherical model, the distances $\{d_\px\}$ and angles $\{\varphi_\px\}$ are calculated according to the spherical Earth geometry. Please see Appendix \ref{app:truegeom} for detailed derivations.} The simulations are conducted by moving the nadir-pointing beam of half-width $\varphi_{\text{RX}}=1.6^{\circ}$ of the LEO BS at $1000$ km (unless stated otherwise) over a realization of UEs, simulating $1000-5000$ different spatial locations. Because the dense PPP is approximately ergodic on the sphere, the spatial statistics, to a practical degree, represent the ensemble statistics. The received powers are normalized so that a UE at the beam center has received power $1$ at the LEO BS. Despite the fixed altitude, antenna width, and elevation angle in the simulations, no generality is lost, since the relevant parameter is \textit{the mean number of UEs inside the $-3$ dB footprint}, which will be argued in the analysis of this paper. An extensive description and comparison of the planar and spherical system models is available in \cite{10909705}: When the beam width is narrower than $4^{\circ}$, the planar and spherical models yield approximately the same total interference distributions for elevation angles $\epsilon>35^{\circ}$ for the LEO altitudes $h\leq 2000$ km. This is because for narrow beams, all relevant UEs have approximately equal spatial path losses. Hence, the statistically relevant ``path loss'' is due to antenna gain, which the planar model can approximate accurately. This paper further validates the feasibility of the planar system model in approximating the spherical model.

\textcolor{black}{Other than the geometrical modeling and the ITU-R antenna pattern, the spherical model is equivalent to the planar model.}

\begin{rem}
  For convenience, with slight abuse of notation, we will use the notation $I_{\textit{o}}$ for both the theoretical definition and the simulated values in the spherical Earth model (thus omitting the notation $I_{\textit{o}_E}$). We trust the context makes it clear whether we refer strictly to \eqref{eq:totalinterferenceat0} or \eqref{eq:mathringptot}. 
\end{rem}
\section{ANALYSIS}
\label{sec:analysis}

Let us first derive the beam width scaling that preserves the cell sizes. Since we want the distance between the beam center and the distance $\|x_{\text{RX}}\|$ to the $-3$ dB edge to be equal for all elevation angles, the area-preserving scaling of $\varphi_{\text{RX}}$ can be solved from the system of equations
\begin{equation}
  \begin{cases}
    D_{h,0}\|x_{\text{RX}}\|= \overline{\varphi^{}}_{\text{RX}}\\
    D_{h,\epsilon}\|x_{\text{RX}}\|=\varphi^{}_{\text{RX}},
  \end{cases}
\end{equation}
which yields
\begin{equation}\label{eq:preserveskappa}\varphi_{\text{RX}}  = \varphi_{\text{RX}}(\epsilon)=\overline{\varphi^{}}_{\text{RX}}\sin^2(\epsilon), \end{equation}
where $\overline{\varphi^{}}_{\text{RX}}$ is the nadir-pointing beam width.

  \begin{cor}[Density of the GP]
  \label{thm:densityoftheGP}
  The GP with a deterministic $H = 1 $ is an inhomogeneous PPP on $(0,1)\ni t$ with the density
  \begin{equation}
    \lambda_{\mathcal{G}_1}(t) \triangleq \tilde{\kappa}/t,
  \end{equation}
  where the GP parameter $\tilde{\kappa} = {\kappa}/\log(2)$ and
  \begin{equation}
    \label{eq:kappa}
    {\kappa} \triangleq \pi \lambda \left(\frac{\varphi_{\textup{RX}}}{D_{h,\epsilon}}\right)^2 = \pi \lambda \left(\frac{\varphi_{\textup{RX}}h}{\sin^2(\epsilon)}\right)^2 = \pi \lambda \left(\overline{\varphi}_{\textup{RX}}h\right)^2
  \end{equation}
  is the average number of UEs inside the $-3$ dB cells. Furthermore, for general $H$, with the support in $(-\infty, \infty) \ni t$
    \begin{equation}
    \label{eq:gaussprojdens}
    \lambda_{\mathcal{G}}(t)\triangleq\frac{ \tilde{\kappa}  F_H(t)}{ t}\Big|_{t\in \R_+}-\frac{\tilde{\kappa} (1-F_H(t))}{t}\Big|_{t\in \R_-},
    \end{equation}
    where $F_H(\cdot)$ is the complementary CDF (CCDF) of the RV $H$. With the random access, the density is simply multiplied by $p \in (0,1)$.
  \begin{proof}
    See \cite{10909705}[Lemma 1] for the proof of $\lambda_{\mathcal{G}_1}(\cdot)$ and for the interpretation of $\kappa$, based on simple geometry. The density of general $\mathcal{G}$ is encompassed in the proof of Lemma \ref{cor:GPPGFL}.
  \end{proof}
  \end{cor}
  As shown in the corollary, the \textit{expected number of UEs in the $-3$ dB cells}, $\kappa$, defines the interference statistics of the beam alone. In this sense, the theory is invariant w.r.t. the elevation angle, altitude, and antenna width. As long as we operate within the LEO orbits and elevation angles larger than $35^{\circ}$, this is also an accurate description of the spherical model. Furthermore, since with the cell area-preserving mapping, $\kappa$ is independent of the elevation angle, i.e., of the beam, all beams have the same statistics, and it is meaningful to talk about ``the typical beam''.

  \begin{lem}[PGFL of the GP]
    \label{cor:GPPGFL}
  Let $f(\cdot): \R \rightarrow [0,1]$, s.t. $f(x)\rightarrow 1$ as $|x| \rightarrow \infty$. The PGFL of $\mathcal{G}$ is 
  
  \begin{align}
    \label{eq:PGFLofGP}
    &\mathfrak{G}_{\mathcal{G}}(f)\triangleq \nonumber\\
    &\mathbb{E}\left(\prod_{\px \in \mathcal{G}} f(\px)\right)=\exp\left\{-\int_{-\infty}^{\infty}\left(1-f(x)\right)\lambda_{\mathcal{G}}(x)dx \right\}.
  \end{align}
  \begin{proof}

    Multiplying each $\px \in \mathcal{G}_1 \subset (0,1)$ by the i.i.d. $H_{\px}$, the probability kernel \cite{HAL1}[Thm. 1.3.9 (Displacement Theorem)] is $\rho(x,y) = f_H(y/x)/x$, where $f_H(\cdot)$ is the probability density function (PDF) of $H$. We have

    \begin{align}
      \label{eq:proofofthePGFL}
      &\mathbb{E}\left(\prod_{\py \in \mathcal{G}} g(\py)\right)  = \mathbb{E}_{\mathcal{G}_{1}}\left(\int_{-\infty}^{\infty} g(y)\prod_{\px \in \mathcal{G}_{1}}\rho(\px,y)dy\right) \nonumber \\
      &= \mathbb{E}\left(\prod_{\px \in \mathcal{G}_{1}}\left(\int_{-\infty}^{\infty} g(y)\rho(\px,y)dy\right) \right) \nonumber \\
      &\overset{(a)}{=}\exp\left\{-\tilde{\kappa} \int_{0}^{1}\left(1-\int_{-\infty}^{\infty}g(y)\rho(t,y)dy\right)/t dt \right\} \nonumber \\
      &\overset{(b)}{=}\exp\left\{-\tilde{\kappa}  \int_{-\infty}^{\infty}\left(1-g(y)\right) \int_{0}^{1}\rho(t,y)/t dt   dy\right\} \nonumber \\
      &\overset{}{=}\exp\left\{-\tilde{\kappa}  \int_{-\infty}^{\infty}\left(1-g(y)\right) \int_{0}^{1}f_H(y/t)/t^2 dt dy\right\} \nonumber \\
      &\overset{(c)}{=}\exp\left\{-\int_{-\infty}^{\infty}\left(1-g(y)\right) \lambda_{\mathcal{G}}(y) dy\right\},
    \end{align}
   where in (a) we use the PGFL of $\mathcal{G}_1$ (see \cite{10909705}[Eq. (15)]). In (b), we note that $\int_{-\infty}^{\infty}\rho(t,y)dy =1$, and switch the integration order. (c) follows by partial integration of the inner integral separately for $y<0$ and $y>0$. 
  \end{proof}
  \end{lem}

   \begin{cor}
     \textcolor{black}{Let $W=0$.} The SIR of the nearest UE signal at the typical Gaussian LEO BS beam is given by 
     \begin{align}
       \label{eq:hatSIR}
       & {\textup{SIR}}_{} =  \frac{{H_{\px_0}G(\|\px_0\|)}}{{I_\textit{o}-G(\|\px_0\|)}}\overset{}{=} \left(\frac{{I_\textit{o}-G(\|\px_0\|)}}{{H_{\px_0}G(\|\px_0\|)}}\right)^{-1} \nonumber \\
       &=\left(\frac{\sum\limits_{\px \in \Phi_{} \setminus \{\px_0\} }H_{\px}G(\| \px\|)}{H_{\px_0}G(\| \px_0\|)} \right)^{-1} \nonumber\\
       &= H_{\px_0}/\sum_{\px \in \mathcal{G}_{1}} H_{\px}\px=H_{\px_0}/I_\textit{o}.
     \end{align}
     \begin{proof}
       The result follows by conditioning $\px_0=\textit{o}$ and applying Slivnyak's theorem on the ratio $G(\|\px\|)/G(\|\px_0\|)$. See details in \cite{10909705}[Lemma 1]. 
     \end{proof}
   \end{cor}

   \begin{rem}
     Indeed, the SIR representation \eqref{eq:hatSIR} is equivalent to conditioning on the nearest UE at $\textit{o}$ (when $G(\|\px_0\|) = 1$). Hence, even though independent of the nearest UE location, \eqref{eq:hatSIR} is always guaranteed to describe the SIR statistics of the nearest UE. This peculiar property is due to the Gaussian form of the path loss (gain) $G(\cdot)$.     
   \end{rem}

   \section{SPATIO-TEMPORAL INTERFERENCE CORRELATION}
   \label{sec:spatiotemporal}
   The first and second-order statistics of the total interference are up to constants determined by the corresponding statistics of the typical signal $H$:
   
      \begin{thm}[Mean and variance of the total interference]
    \label{thm:meanandthevarianceofthetotalI}
    Assume that the mean and the average second power $|\mathbb{E}(H)|,\mathbb{E}(H^2)<\infty$, respectively. The mean and the variance of the total interference are
    \begin{align}
        \label{eq:expectedI}      &\mathbb{E}(I_{\textit{o}})=\mathbb{E}\left(\sum_{\px \in {\mathcal{G}}}  \px + W\right)\overset{(a)}{=} \int_{-\infty}^{\infty} x \lambda_{\mathcal{G}}(x)  dx +W \nonumber \\ &= \tilde{\kappa}\mathbb{E}(H) + W,     \\
      &\textup{var}(I_{\textit{o}}) =\textup{var}\left(\sum_{\px \in {\mathcal{G}}}  \px \right)\overset{(b)}{=}\int_{-\infty}^{\infty} x^2 \lambda_{\mathcal{G}}(t)  dx \nonumber \\
        &= \tilde{\kappa} \left(\int_{0}^{\infty} xF_H(x)  dx -\int_{-\infty}^{0} x(1-F_H(x))  dx \right)  \nonumber \\        
        &= \tilde{\kappa} \mathbb{E}(H^2)/2,        \label{eq:varI}
    \end{align}
        respectively. With the random access, the mean and the variance are simply multiplied by $p \in (0,1)$.
    \begin{proof}
      The integral identities (a) and (b) for the mean and the variance of the sum of the PPP can be found in \cite[Cor. 4.8]{alma998070634406526}.
    \end{proof}

      \end{thm}
      Without loss of generality, assume normalized $\mathbb{E}(H)=1$ and $W=0$. Denote $c \triangleq \|u-v\|$ where $u,v\in \R^2$ are the beam locations. Let $\epsilon_u$ and $\epsilon_v$ be the elevation angles at $u$ and $v$, respectively. The spatio-temporal correlation coefficient of the interferences $I_u(k)$ and $I_v(l)$ at time instances $k \neq l$ is obtained using Campbell's theorem and the second-order product density of a PPP:
      \begin{align}
        \label{eq:spati-temporl}
        &\zeta(u,v) = \zeta(\|u-v\|)=\zeta(c) \nonumber \\
        &\triangleq\mathbb{E} \left( (I_{u}(k)-\mathbb{E}(I_\textit{o})) (I_{v}(l)-\mathbb{E}(I_\textit{o}))\right)/\textup{var}(I_{\textit{o}})  \nonumber\\
        &= \mathbb{E}\Bigg(\sum_{\px \in \Phi(k)}H_{\px}(k) 2^{-(D_{h,\epsilon_u} \|\px-u \|)^2/\varphi_{\textup{RX}}(\epsilon_u)^2} \nonumber  \times \nonumber\\
        &\hspace{1.3cm}\sum_{\py \in \Phi(l)}H_{\py}(l) 2^{-(D_{h,\epsilon_v} \|\py -v\|)^2 /\varphi_{\textup{RX}}(\epsilon_v)^2}\Bigg)/\textup{var}(I_{\textit{o}}) \nonumber
      \end{align}
      \begin{align}
        &=\frac{p^2\mathbb{E}(H)^2 \lambda}{p\tilde{\kappa}\mathbb{E}(H^2)/2} \int_{\R^2} 2^{-\|x \|^2/(\overline{\varphi}_{\textup{RX}}h)^2-\|x -c\|^2/(\overline{\varphi}_{\textup{RX}}h)^2} dx  \nonumber \\
        &= \frac{p\lambda}{\tilde{\kappa}\mathbb{E}(H^2)/2}\frac{\tilde{\kappa}}{2\lambda} \exp\left\{-c^2/(2\overline{\varphi}^2_{\text{RX}}h^2) \log(2)\right\}\nonumber \\
        &= \frac{p}{\mathbb{E}(H^2)}\exp\left\{-c^2/(2\overline{\varphi}^2_{\text{RX}}h^2) \log(2)\right\}.
      \end{align}
      \textcolor{black}{This is the total interference correlation coefficient between two beams separated by a distance $c$, with beam widths scaled by \eqref{eq:preserveskappa} according to the beam elevation angle. Another formulation and interpretation is when a beam of width $\varphi_{\text{RX}}$ moves a distance $c$ (as the satellite moves):}
      \begin{equation}
        \label{eq:singlebeamcorr}
        \zeta(c) = \frac{p}{\mathbb{E}(H^2)}\exp\left\{-(cD_{h,\epsilon})^2/(2\varphi_{\textup{RX}}^2)\log(2)\right\},
      \end{equation}
      which is the same as \eqref{eq:spati-temporl} if we use the footprint area-preserving scaling of $\varphi_{\text{RX}}$ for each $\epsilon$. Note that \eqref{eq:singlebeamcorr} reflects the Gaussian-form antenna gain \eqref{eq:antennagain} of a beam width $\varphi_{\text{RX}}$. The correlation coefficients are independent of $\lambda$ and $\tilde{\kappa}$.

      \begin{figure}[ht]
        \centering
        \includegraphics[width=\linewidth]{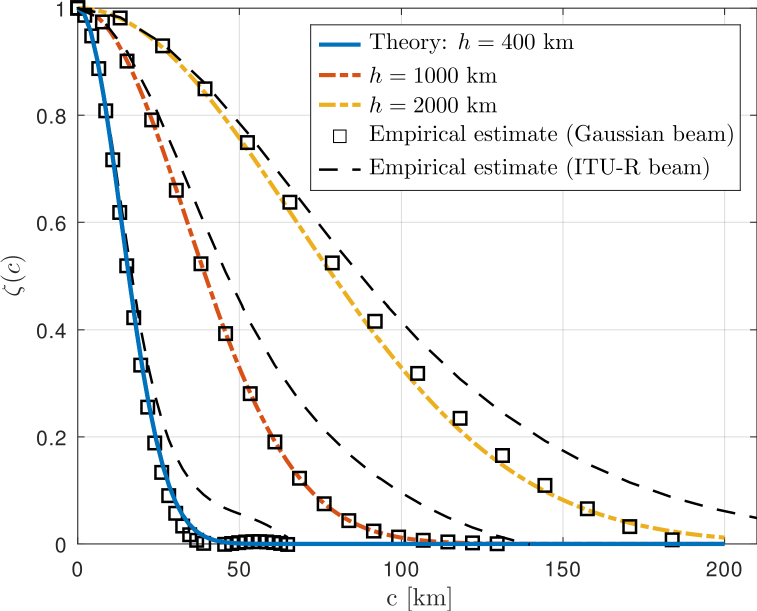}
        \caption{The spatial total interference correlation at altitudes $h \in \{400, 1000, 2000\}$ km. Deterministic fading $H = 1$, and the transmission probability $p=1$. }
        \label{fig:spatialcorrelation}
      \end{figure}
      
\textcolor{black}{Fig. \ref{fig:spatialcorrelation} depicts the theoretical and simulated spatial interference correlation coefficients at various altitudes between beams separated by $c$, or equivalently, for a single beam as the satellite moves a distance of $c$. The Gaussian and ITU-R beams \eqref{eq:ITUpattern} were used in the simulations for the empirical estimates. The simulated Gaussian correlations closely match the Gaussian form. The ITU-R estimates are biased, especially for larger $c$, due to the difference between the ITU-R and Gaussian gains. However, the correlation functions are near-Gaussian. Hence, we can infer that the sidelobe interference variance is small compared to the main-lobe variation since it has limited effect on the correlation coefficient. The Gaussian main-lobe component is statistically the most significant.}

     Since the sidelobes have a limited effect on the correlation coefficient, we put forth the following proposition.
      \begin{prop}[sidelobe modeling]
        \label{prop:sidelobemodeling}
        Interference from the antenna beam sidelobes can be approximated by a constant variable $W>0$.
      \end{prop}
      $W$ depends on the domain and density of interfering UEs, the altitude of the satellite, the width of the antenna beams, and the elevation angle. In this work, we do not derive an explicit formula, but determine $W$ empirically through simulations.


      \section{TEMPORAL CORRELATION OF LINK OUTAGES }
      \label{sec:tempcorlinkout}
      Let $A_k$ denote the event that the nearest UE SIR exceeds the threshold $\theta$ at time instant $k$ at the typical beam. Similarly to \eqref{eq:hatSIR}, one can define the SIR for a general $W\geq 0$ as
      \begin{equation}
        \textup{SIR} = \frac{H_{\px_0}(k)}{\sum_{\px \in \mathcal{G}_1(k)} H_{\px}(k) \px + W/\px_0} >\theta,
      \end{equation}
      where $W$ is divided by the largest gain $\px_0$ (the gain of the nearest UE).

      Crucially, the SIR distribution with Rayleigh fading closely approximates that with general Nakagami-$m$ fading, where $m$ is the Nakagami fading parameter, in the simple coverage region $\theta \geq 1$, and it is pessimistic for $\theta <1$ \cite{10909705}[Figs 6 and 7]. Hence, for analytical lubrication, we use $H_{\px}\sim\text{exp}(1)$. Similar analytical methods apply to Nakagami fading in the following analysis. However, the expressions become unnecessarily complicated, given our focus on qualitative insight rather than numerical precision. In the simple coverage region, the following analysis also applies to Rice-$K$ fading, where $K$ is the Rician parameter.

      The joint probability of the sequential events $A_k$, $A_l$, $l\neq k$, occurring is as follows.
   \begin{align}
     \label{eq:jointprobAkAl}
     &\mathbb{P}(A_k,A_l) \nonumber\\
     &=\mathbb{P}\Bigg(H_{\px_0}(k) > \theta \left(\sum_{\px \in \mathcal{G}_1(k)} H_{\px} (k)\px  +\frac{W}{\px_0(k)}\right), \nonumber \\
     &\hspace{2cm}H_{\px_0}(l) > \theta \left(\sum_{\px \in \mathcal{G}_1(l)} H_{\px} (l) \px + \frac{W}{\px_0(l)} \right) \Bigg) \nonumber \\
     &=\mathbb{E}_{H,\mathcal{G}}\left( \prod_{\px \in \mathcal{G}_1(k)} \exp\left\{-\theta  H_{\px} (k) \px \right\} \right) \times\nonumber \\
     &\hspace{1cm}\mathbb{E}_{H,\mathcal{G}}\left(\prod_{\px \in \mathcal{G}_1(l)}\exp\left\{-\theta  H_{\px} (l) \px \right\} \right) \times \nonumber \\
          &\hspace{1cm}\mathbb{E}\left(e^{- \theta W/\px_0(k)}\right)\mathbb{E}\left(e^{- \theta W/\px_0(l)} \right) \nonumber \\
     & \overset{(a)}{=}\mathbb{E}_{\mathcal{G}}\left( \prod_{\px \in \mathcal{G}_1}\left(\frac{p}{1+\theta \px} +1-p\right)^2\right) \mathcal{L}_{1/\px_0}(s)^2 \nonumber \\
     &\overset{(b)}{=} \exp\left\{ -\tilde{\kappa} \int_0^1 \left(1- \left(\frac{p}{1+\theta r} +1-p\right)^2\right)/r dr\right\} \times \nonumber \\
     & \hspace{1cm}\tilde{\kappa}^2E_{\tilde{\kappa}+1}( \theta W)^2
     \nonumber\\
     &=e^{-p^2\theta \tilde{\kappa}/(1+\theta)}(1+\theta)^{p\tilde{\kappa}(p-2)} (p\tilde{\kappa})^2E_{p\tilde{\kappa}+1}(W \theta)^2,
   \end{align}
   where $E_{\tilde{\kappa}+1}(\cdot)$ is the generalized exponential integral.   In (a), we  average over the fading RVs (Laplace transform of the exponential RV) and the slotted ALOHA random access. (b) is the PGFL of the GP and the Laplace transform $\mathcal{L}_{X}(s) \triangleq \mathbb{E}(e^{-sX})$ of the inverse largest gain $X=1/\px_0$ (see Appendix \ref{appendix:largestgain}).

   Similarly, one can derive the SIR as the \textit{Lomax distribution} tempered by the exponential integral:
   \begin{equation}
     \label{eq:SIRLomaxdist}
     \mathbb{P}(A_k)= (1+\theta)^{-p \tilde{\kappa}}p\tilde{\kappa}E_{p\tilde{\kappa}+1}(W \theta),
     \end{equation}
   and the conditional probability is
   \begin{align}
     &\mathbb{P}(A_k|A_l)= \frac{\mathbb{P}(A_k,A_l)}{\mathbb{P}(A_l)} \nonumber \\
     &= e^{-p^2\theta \tilde{\kappa}/(1+\theta)}(1+\theta)^{p^2\tilde{\kappa}}(1+\theta)^{-p\tilde{\kappa}}p\tilde{\kappa}E_{p\tilde{\kappa}+1}(W \theta).
   \end{align}
   We have $\mathbb{P}(A_k|A_l) > \mathbb{P}(A_k)$ for all $p \in (0,1]$. In a single access network ($p=1$), without the sidelobes (when $W=0$ and  $\tilde{\kappa} E_{\tilde{\kappa}+1}(W\theta)=1$), in the limit,
   \begin{equation}
      \mathbb{P}(A_k|A_l) \rightarrow e^{-\tilde{\kappa}} \textup{ as } \theta \rightarrow \infty,
   \end{equation}
   which clearly differs from the corresponding limit of $\mathbb{P}(A_k)$ for small $\tilde{\kappa}$. The probability of connecting at the next time instance is higher if the UE-LEO BS link is already connected at the current time. This indicates that some LEO BS beams may be in outage during a use period, while others provide a consistently good connection. \textcolor{black}{The SIR correlation is determined through the correlation coefficient of the total interference. This is because the SIR can be characterized by the inverse of the total interference (recall (17)). Hence, the outage correlation vanishes as the LEO BSs move outside the region where the total interference is (to a practical degree) correlated (see Fig. 2).} On the other hand, the presence of sidelobe interference (or noise) reduces spatio-temporal correlation (while deteriorating the SIR): For any $W>0$, $\mathbb{P}(A_k|A_l)=\mathbb{P}(A_k)=0$, as $\theta \rightarrow \infty$, and for small $\theta$, the difference in the conditional and unconditional probabilities depends on the magnitude of $W$. The spatial clustering of outages in the single access network can be mitigated by a random access network $p \in (0,1)$.

   The mean and the variance of the SIR without the sidelobes ($W=0$) are
   \begin{align}
     \label{eq:SIRmean}
     \mathbb{E}(\textup{SIR})&= \int_{0}^{\infty}(1+t)^{-p\tilde{\kappa}} dt = \frac{1}{p\tilde{\kappa}-1},\textup{ for } \tilde{\kappa}>1 \\
     \textup{var}(\textup{SIR}) &= 2 \int_{0}^{\infty}t(1+t)^{-p\tilde{\kappa}} dt - \left(\int_{0}^{\infty}(1+t)^{-p\tilde{\kappa}} dt\right)^2 \nonumber\\
     &=   \frac{p^2\tilde{\kappa}^2}{(p\tilde{\kappa}-2)(p\tilde{\kappa}-1)^2}  , \textup{ for }\tilde{\kappa }>2  \label{eq:SIRvariance},
   \end{align}
   respectively. The mean and variance are finite for $p\tilde{\kappa}>1$ and $p\tilde{\kappa}>2$, respectively; for $p\tilde{\kappa}<2$, the variation in the SIR is significant across the beams. This notion is made precise in the SIR meta distribution analysis in \cite{10909705}[Sec. III]. With sidelobes or noise ($W>0$), the mean and the variance are always finite. However, they have complicated expressions. 

         \begin{figure}[ht]%
        \centering
        \subfloat[Total interference.]{{\includegraphics[width=\linewidth]{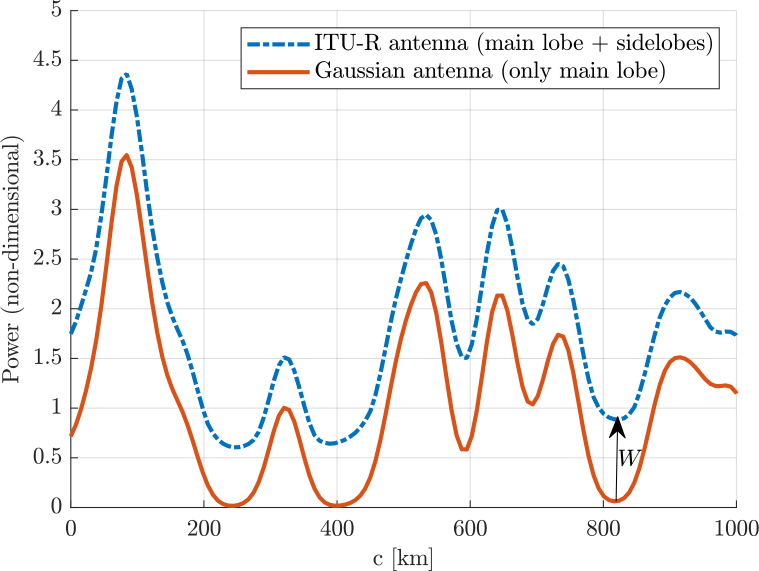}}
          \label{fig:totalinterfence}
        } 
        \qquad
        \subfloat[SIR of the nearest UE.]{{\includegraphics[width=\linewidth]{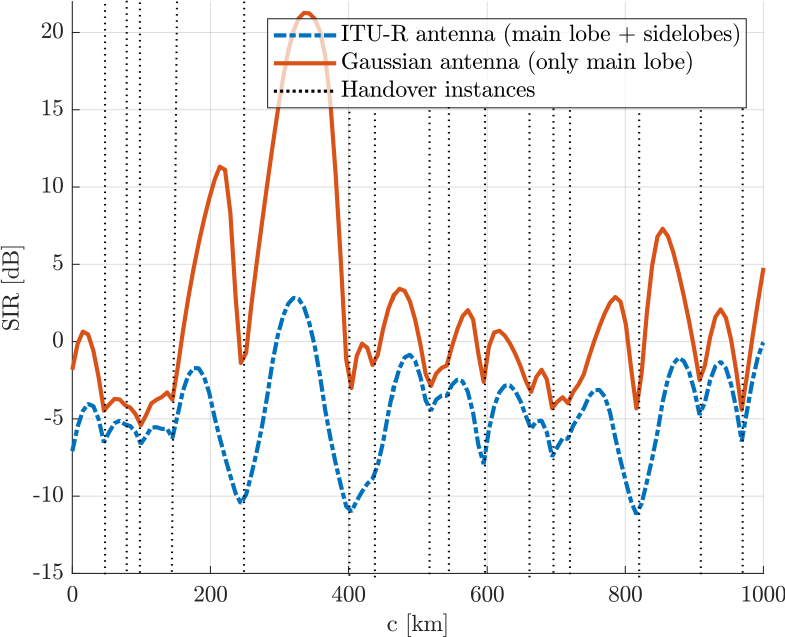}}     \label{fig:nearestUESIR}} 
        \caption{The total interference and the nearest UE SIR in a single access network ($p=1$) \textcolor{black}{at a beam} as the satellite moves a distance $c$ (equivalently, \textcolor{black}{between two beams separated by $c$}). In the simulations, the LEO BS is at $h=1000$ km and $\overline{\varphi}_{\text{RX}}=1.6^{\circ}$, which corresponds to $\overset{\sim}{\kappa}=1$. The sidelobe power $W=0.65$. Fading is deterministic; $H=1$.}
         \end{figure}

         In Figs. \ref{fig:totalinterfence} and \ref{fig:nearestUESIR}, the simulated total interference and the nearest UE SIR are plotted for the ITU-R and Gaussian antennas in the line-of-sight (LoS) channel (deterministic $H$). The handover instances when the served (and nearest) UE ceases to be the nearest UE in the beam are marked in Fig. \ref{fig:nearestUESIR} by a dotted line. Notably, the total interference for the ITU-R antenna is a mere translation of the interference with the Gaussian antenna, further validating Proposition \ref{prop:sidelobemodeling}. Furthermore, the SIRs are correlated between the antennas. The GP parameter $\tilde{\kappa}=1$ is in the infinite-mean and undefined-variance SIR regime with the Gaussian antenna. For the ITU-R antenna, the variance is relatively small ($\textup{var}(\textup{SIR}) \approx -16$ dB), but it still exhibits much greater variation than for the larger $\tilde{\kappa}$. The total interference (defined, recall, as the total received power from all UEs) is slightly correlated with the corresponding antenna SIR, suggesting a minimum-complexity method for channel-state estimation. For the Gaussian antenna, this correlation is always negative. For the ITU-R antenna, the correlation is positive for small GP parameters ($\tilde{\kappa}<4$) and negative for large parameters. This is because for small $\tilde{\kappa}$, high total interference entails a high probability that the served UE is in the main lobe beam with no interferers in the main beam. On the other hand, for large $\tilde{\kappa}$, high total received power likely also entails considerable main lobe interference. For the random series in Figs. \ref{fig:totalinterfence} and \ref{fig:nearestUESIR}, the Pearson cross-correlation coefficients $\mathbb{E}((\textup{SIR}-\mathbb{E}(\textup{SIR}))(I_{\textit{o}}-\mathbb{E}(I_{\textit{o}})))/\sqrt{\textup{var}(\textup{SIR})\textup{var}(I_{\textit{o}})}$ between the SIR and the total interference are $0.12$ and $-0.04$ for the ITU-R antenna and the Gaussian antennas, respectively.

         Recall that the SIR values approximately follow the Lomax distribution statistics in the simple coverage region, and the Lomax model is pessimistic otherwise. The interference power statistics approximately follow the gamma distribution, which will be argued in the following.

   \begin{cor}[Gamma distribution model for the interference]
     \label{ex:gammdistmodel}
        Assume Rayleigh fading and $W=0$. The Laplace transform of the total interference is
        \begin{equation}
          \label{eq:gammamodelrayray}
          \mathcal{L}_{I_{\textit{o}}}(s) = \mathbb{E}(e^{-I_{\textit{o}}s}) = (1+s)^{-p \tilde{\kappa}},
        \end{equation}
        which is the Laplace transform of the gamma distribution with the shape parameter $p \tilde{\kappa}$.
        \begin{proof}
          The result is a direct consequence of \eqref{eq:SIRLomaxdist}. Namely, assuming $H \sim \exp(1)$, for $s \geq 0$,
          \begin{align}
            &\mathbb{P}(\textup{SIR} > s) = \mathbb{P}\left(H_{\px_0}(k) >s \sum_{\px \in \mathcal{G}_1(k)} H_{\px}(k) \px \right) \nonumber\\
            &=\mathbb{E}_{H,\mathcal{G}_1}\left( \exp\left\{-s \sum_{\px \in \mathcal{G}_1(k)} H_{\px}(k) \px \right\}\right)  = \mathbb{E}(e^{-sI_{\textit{0}} }),
          \end{align}
          which can be extended to the non-negative complex half-plane by the analytical continuation. 
        \end{proof}
   \end{cor}
   From the corollary, we put forth a gamma distribution interference approximation model for Rician fading by determining the shape and scale parameters via second-order moment matching. Furthermore, with the sidelobes, a gamma distribution shifted by $W$ is feasible.

   In Figs. $\ref{fig:gaussI0}$ and $\ref{fig:ITURI0}$, the total empirical and theoretical interference distributions are plotted in the LoS channel with the GP parameter $\tilde{\kappa} =1$, and in $\ref{fig:gaussI0tkappa5}$ and $\ref{fig:ITURI0tkappa5}$ for the Rayleigh channel with $\tilde{\kappa} =3.3$. The empirical values are simulated in the spherical model. For the theory, we use the gamma distribution model proposed in Corollary \ref{ex:gammdistmodel} with moment matching to the mean and the variance \eqref{eq:expectedI} and \eqref{eq:varI}, respectively. For the ITU-R beam, we shift the PDF by $W= 0.65$ for $\tilde{\kappa} = 1$, as empirically observed in Fig. \ref{fig:totalinterfence}; whereas for $\tilde{\kappa} = 3.3$ by $W=2.25$, similarly empirically derived.  Recall that since PPP is ergodic, the ensemble statistics follow the spatial statistics. Hence, the gamma distribution model applies to the multi-access network ($p<1$). For $\tilde{\kappa} \rightarrow \infty$, the interference distributions approach a Gaussian distribution in the positive domain. The empirical histograms coincide closely with the theoretical PDFs, and thus the gamma distribution model is a good, tractable candidate for interference modeling in the narrow-beam LEO uplink.

         \begin{figure}[ht]%
           \centering
           \subfloat[Gaussian antenna (only main lobe).]{{\includegraphics[width=\linewidth]{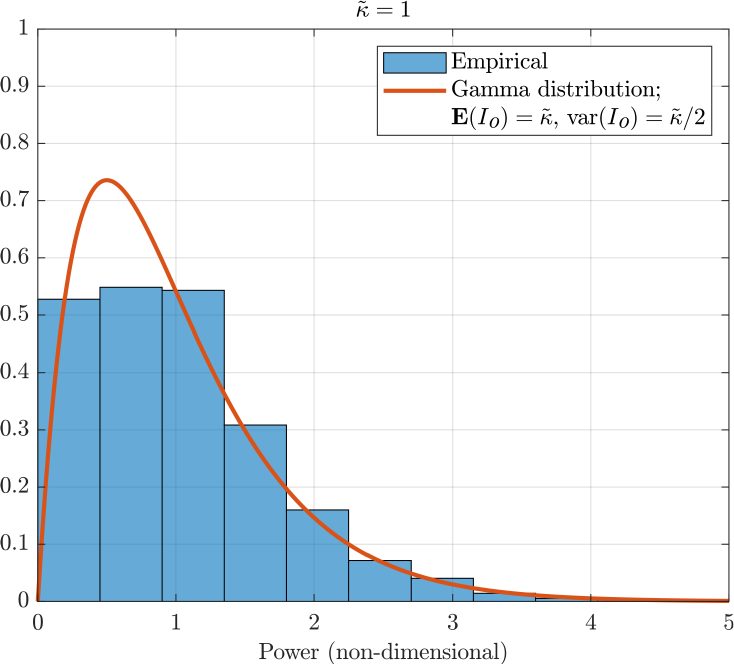}}
               \label{fig:gaussI0}}
           \qquad
           \subfloat[ITU-R antenna (main lobe $+$ sidelobes).]{{\includegraphics[width=\linewidth]{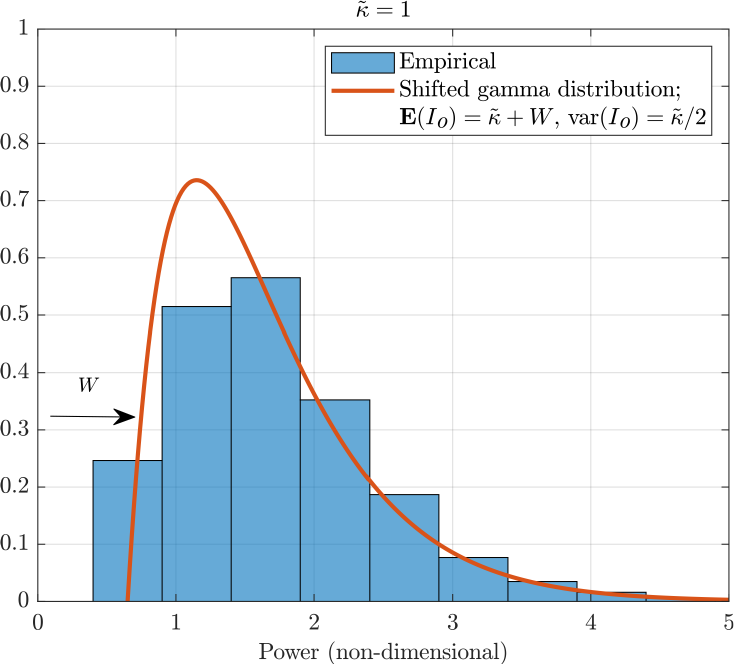}}     \label{fig:ITURI0}} 
           \caption{The simulated LoS channel ($H=1$) total interference PDFs \textcolor{black}{at a beam as the LEO BS moves} over a period of time in a single access network ($p=1$), and the gamma distribution model. The simulation parameters are $\overline{\varphi}_{\text{RX}} = 1.6^{\circ}, h=1000$. The sidelobe power $W=0.65$.}
         \end{figure}

         \begin{figure}[ht]%
           \centering
           \subfloat[Gaussian antenna (only main lobe).]{{\includegraphics[width=\linewidth]{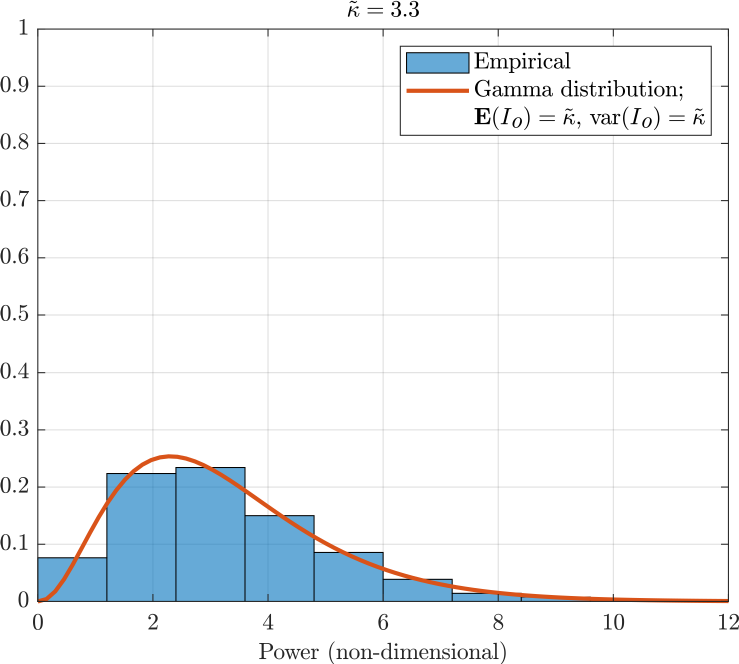}}
               \label{fig:gaussI0tkappa5}}
           \qquad
           \subfloat[ITU-R antenna (main lobe $+$ sidelobes).]{{\includegraphics[width=\linewidth]{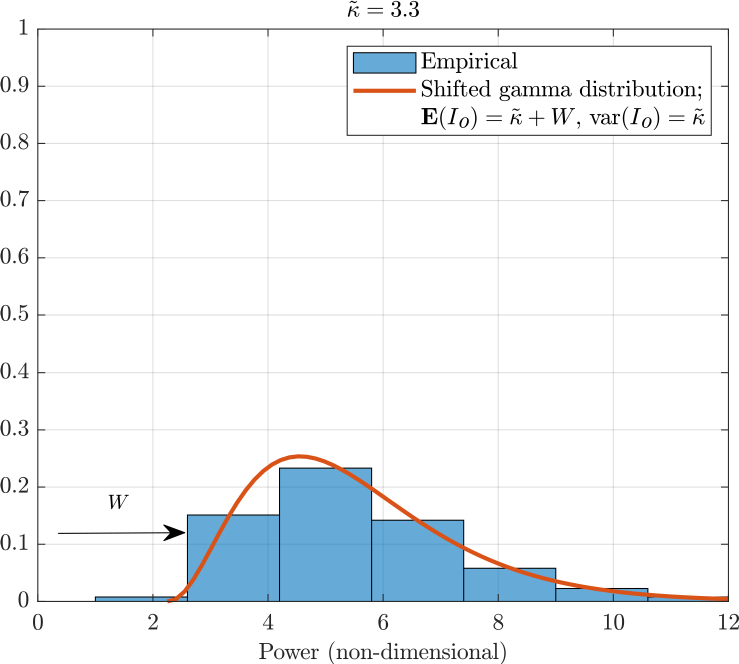}}     \label{fig:ITURI0tkappa5}} 
           \caption{The simulated Rayleigh channel ($H=1$) total interference PDFs \textcolor{black}{at a beam as the LEO BS moves} over a period of time in a single access network ($p=1$), and the gamma distribution model. The simulation parameters are $\overline{\varphi}_{\text{RX}} = 1.6^{\circ}, h=1000$. The sidelobe power $W=2.25$.}
         \end{figure}

   \subsection{MAXIMAL AVERAGE THROUGHPUT WITHOUT SPATIAL CORRELATION}
   With the Gaussian beams, the average bandwidth-normalized Shannon throughput of the LEO multiple-access network is given by \cite{10909705}[Eq. (36)] 
   \begin{equation}
     \label{eq:throughput}
     \mathbb{E}(\log(1+\textup{SIR}))= 1/(p \tilde{\kappa}).
   \end{equation}
   Furthermore, in an interference-plus-noise-limited channel (or a channel with sidelobe interference), the optimal throughput is given for $p\tilde{\kappa}\approx 1$ for moderate $W$. This means, on average, $p{\kappa}\approx\log(2) \approx 0.7$ effective UEs in a $-3$ dB beam footprint area, meaning that the served UE is inside the $-3$ dB beam footprint with probability $1/2$ and in the main lobe footprint with probability $9/10$.

   Recalling \eqref{eq:SIRmean} and \eqref{eq:SIRvariance} for $p=1$, in a single access network, especially with small sidelobe interference, due to significant spatial variation in interference and SIR, some LEO BS beams serve with high average throughput, while others serve with low average throughput. However, using slotted ALOHA random access with a suitable transmission probability $p$, we can mitigate spatial correlation among the SIRs and improve throughput consistency compared to the single access network, while preserving average throughput. Namely, for given $\rho\geq1$, set $p \triangleq 1/\rho$, and set the new GP parameter $\tilde{\kappa}_{\text{ALOHA}} \triangleq \rho \tilde{\kappa}$. Then the throughput $\mathbb{E}(\log(1+\textup{SIR}))=1/(p \tilde{\kappa}_{\text{ALOHA}}) = 1/\tilde{\kappa}$, but it is easy to see that 
 \begin{equation}
   \label{eq:limitconditionalprob}
    \mathbb{P}(A_k,A_l)= \mathbb{P}(A_k)\mathbb{P}(A_l)
  \end{equation}
 in the limit $\rho \rightarrow \infty$. That is, the distributions of the SIR and the throughput of sequential transmissions are independent: Conditioning on any constellation configuration, the temporal average throughput, averaged over a use period, at each LEO BS beam is equivalent in the limit $\rho \rightarrow \infty$ in the multi-access network while preserving the throughput of the single access network. As a trade-off, the transmission window of the served UEs is reduced, each being a fraction $p=1/\rho$ of the use period.

 \begin{example}[An equivalent slotted ALOHA network]
   Let $\varphi_{\textup{RX}}=\overline{\varphi}_{\textup{RX}} = 0.0278$, $\epsilon =\pi/2$, $h=1000 \textup{ km}$, and $\lambda = 4.1 \times 10^{-4}/\textup{ km}^2$. The GP parameter is
   $$\tilde{\kappa} = \pi \lambda \left(\frac{\varphi_{\textup{RX}}}{\sin^2(\epsilon)/h} \right)^2 = \pi 4.1 \times 10^{-4} (0.0278 \times 1000)^2 \approx 1,$$
   which corresponds to the average network throughput $1/(p \tilde{\kappa})=1$. To stabilize the performance of each LEO BS beam while preserving the average throughput, we can make the number of co-channel transmissions three times denser, $ \lambda=12.3 \times 10^{-4}/\textup{ km}^2$, and set the transmission probability to $p=1/3$. This corresponds to a $3$-tier multi-access network.
 \end{example}

\begin{rem}[Interpretation of the GP parameter as the cell size]
  One can see from the void probability of the PPP (see Appendix \ref{appendix:largestgain}) that for $p\tilde{\kappa}=1$, a fraction $9/10$ of served UEs are in the $-10$ dB beams. On the other hand, for $p\tilde{\kappa}=3.3$, the equivalent percentile applies to the $-3$ dB beams. In this sense, $p\tilde{\kappa} \in \{1,3.3\} $ can be considered to correspond to the cell sizes of $-10$ dB (large) and $-3$ dB (small) footprints, respectively. 
\end{rem}

In Figs. \ref{fig:probofccov} and \ref{fig:probofcovITU}, we plot the unconditional and conditional coverage probabilities for the GP parameters $p\tilde{\kappa} \in \{1,3.3\}$ and transmission probabilities $p \in \{1,1/3\}$ for the Gaussian beam and the ITU-R beam, respectively. Recall from \eqref{eq:SIRvariance} that, for the Gaussian beam, $p\tilde{\kappa}=1$ is in the infinite variance SIR region, and for $p\tilde{\kappa}=3.3$, the variance is finite, reflecting smaller variation in the throughput over the beams. Hence, for small cell sizes, the user experience is more consistent than for large cell sizes. One can see that in the multi-access network, increasing $\tilde{\kappa}$ while decreasing $p$ accordingly causes the conditional distribution $\mathbb{P}(A_k|A_l) $ to approach the unconditional distribution $\mathbb{P}(A_k)$. This shows that random access and appropriate transmission probability allow the average throughput at each LEO BS beam over each use period to remain consistent \textit{from the perspective of the LEO BSs}. Furthermore, a crucial observation is that sidelobe interference significantly degrades the SIR.

In Fig. \ref{fig:probofccov}, the error in the simulated conditional coverage probability for the ITU-R beam is because, with $p\tilde{\kappa}=1$, the served UE has a $1/10$ probability of being outside the $-10$ dB beam, which is not accurately captured by the proposed theoretical model for small $\theta$. 

Recall that the analytically derived (theoretical) SIR distribution is for Rayleigh fading: with a line-of-sight component, the unconditional Rayleigh coverage probabilities (i.e., the Lomax distribution or the tempered Lomax distribution) are pessimistic for $\theta <0$ dB and approximate the Rice coverage probabilities for $\theta >0$ dB. The conditional Rayleigh probabilities are pessimistic vis-à-vis the Rician fading.



          \begin{figure}[ht]%
           \centering
           \subfloat[Large cell size]{{\includegraphics[width=\linewidth]{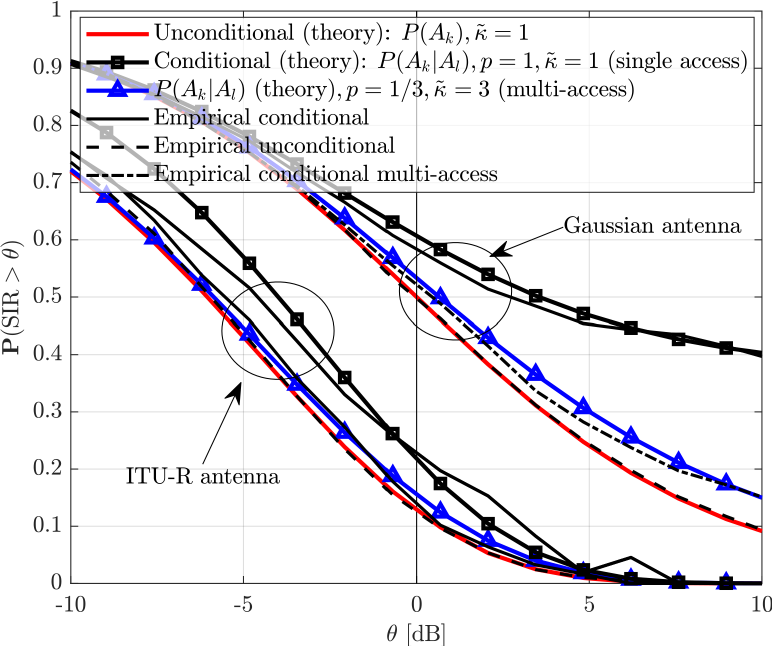}}
               \label{fig:probofccov}}
           \qquad
           \subfloat[Small cell size]{{\includegraphics[width=\linewidth]{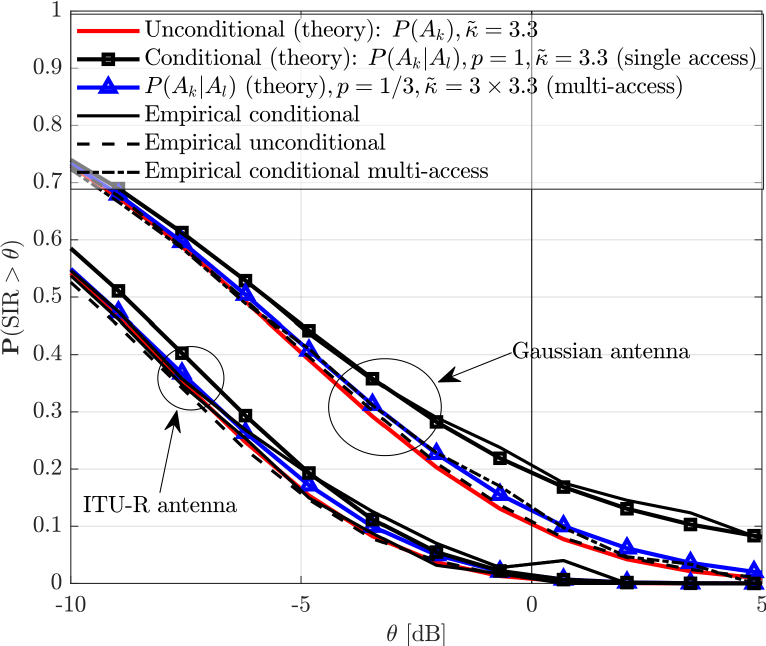}}     \label{fig:probofcovITU}} 
           \caption{The conditional and unconditional coverage probabilities for $p \overset{\sim}{\kappa}_{}  \in \{1,3.3\}$ and the transmission probabilities $p \in \{1,1/3\}$. In the single access network, $p=1$, the GP parameters correspond to $\lambda \in$ $\{4.0 \times 10^{-4},$ $1.3 \times 10^{-3}\}/\textup{ km}^2$, respectively, for $h=1000$ km and $\varphi_{\textup{RX}} = 1.6^{\circ}$. }
         \end{figure}

          \section{CONCLUSIONS}
          \label{sec:conclusions}
   
  To wrap up the observations:
  \begin{enumerate}
  \item The spatial correlation coefficient of interference power at the typical LEO BS beam has a simple Gaussian form.
  \item The interference power distribution is approximately a gamma distribution or a shifted gamma distribution in an Rician fast-fading channel.
  \item In a single access network, if the UE-LEO BS link is in outage during a block, the probability that the next block is in outage is increased compared to a connected link, illustrating spatial clustering of outages during the use periods.
  \item The slotted ALOHA protocol can mitigate the spatial clustering of outages while preserving average network performance.
  \item On the downside, the slotted ALOHA does not improve user fairness: the spatial SIR correlation mitigation applies only from the perspective of the LEO BSs being an average over the served UEs; however, the throughput of individual UEs varies significantly.
  \end{enumerate}


   


  The UE SIR variation is significant due to high UE signal gain variation in the LEO BS beams as the satellites move, especially for large cell sizes. In this regard, slotted ALOHA is a promising solution for equalizing throughputs among LEO BS beams compared to a single access network. For smaller cell sizes, the LEO BS beam throughputs have less spatial clustering; however, intracell and inter-cell interference are increased, which significantly deteriorates the average SIR. More detailed considerations, such as user fairness, achievable throughput from the perspective of the LEO BSs, and related trade-offs, should be further addressed in future explorations of more sophisticated grant-free protocols and non-terrestrial network designs. Furthermore, the optimal cell size should be further investigated, including intercell and intracellular interference and its dependence on practical beamforming and grant protocols. 

  The proposed stochastic geometry framework provides tractable analytical tools for random access, beam management, and other design aspects of LEO networks. For example, the novel gamma distribution interference model and the corresponding closed-form interference correlation function are valuable prior hypothesis classes for machine learning signal processing in sophisticated interference management in narrow-beam LEO networks \cite{9378781}.

   \appendices

   \section{SCALING CONSTANT}
   \label{app:constantD}
   See Fig. \ref{fig:triangle}. We have that $\zeta_{z}= \tan^{-1}(z/h)$. The derivative of $\varphi_x$  around $\textit{o}$ is given approximately by
   \begin{align}
     &\frac{d}{d \| x\|}\varphi_{x} =\frac{d}{dz}\zeta_{z}  = \frac{d\tan^{-1}(z/h)}{dz} = \frac{h}{h^2 + z^2}   \nonumber\\
     &\overset{(a)}{\approx} \frac{h}{h^2 -h^2 + \hat{d}^2_{h,\epsilon}}  \overset{(b)}{=}\frac{h}{ h^2/\sin^2(\epsilon)} = \frac{\sin^2(\epsilon)}{h} = D_{h,\epsilon},
   \end{align}
   where (a) follows from Pythagoras's theorem, and (b) is standard trigonometry.

   \begin{figure}[ht]
     \centering
     \includegraphics[width=0.7\linewidth]{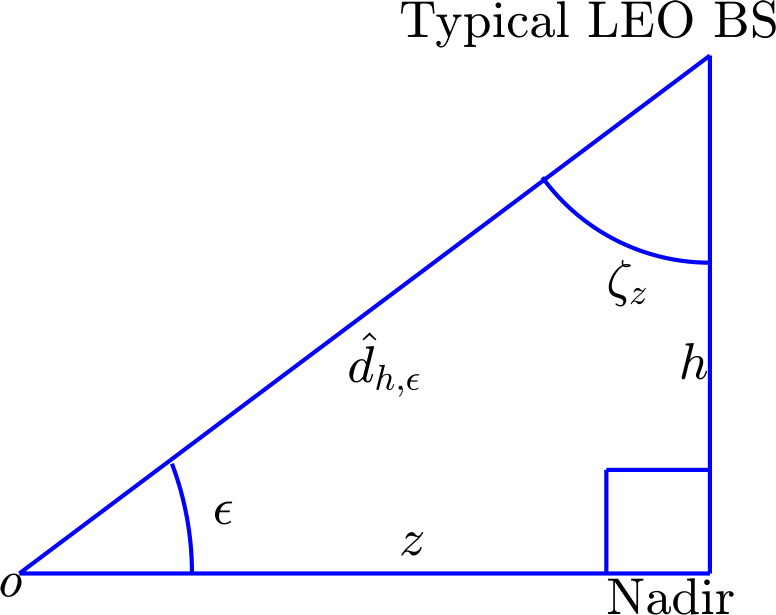}
     \caption{Geometric interpretation of the variables in Appendix \ref{app:constantD} }
     \label{fig:triangle}
   \end{figure}

   \section{Geometry of the spherical system model}
   \label{app:truegeom}

   \begin{figure}[ht]
     \centering
     \includegraphics[width=0.7\linewidth]{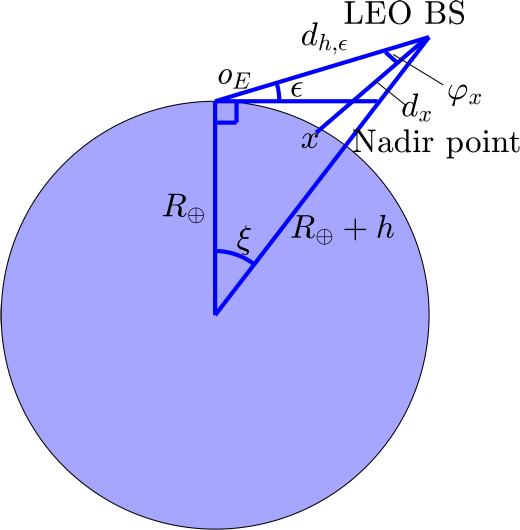}
     \caption{Sketch of the geometry of the spherical model}
     \label{fig:sphericalmodel}
   \end{figure}

   See Fig. \ref{fig:sphericalmodel}. Directly from the law of cosines, we have
   \begin{align}
     \label{eq:hatd}
     &{d}_{h,\epsilon}(\xi) \nonumber\\
     &= \sqrt{\rEarth^2+(\rEarth +h)^2-2\rEarth(\rEarth+h)\cos(\xi)},
   \end{align}
   Furthermore, we may derive the relation between $\epsilon$ and $\xi$. Namely, the law of cosines states that
   \begin{equation}
     \label{yhtaloryhma}
     (\rEarth + h)^2 = {d}_{h,\epsilon}(\xi)^2 + \rEarth^2 - 2 {d}_{h,\epsilon}(\xi) \rEarth \cos(\pi/2 +\epsilon),
   \end{equation}
 which is analytically solvable for $\xi$. 

 For the beam pointed to the nadir, the angle to the boresight can be calculated as
   \begin{equation}
     \label{eq:varphixxxx}
     \varphi_x =\frac{\cos^{-1}\left(d_{h,\epsilon}(\xi_x)^2 + (\rEarth + h)^ 2 - \rEarth^2\right)}{2  d_{h,\epsilon}(\xi_x)(\rEarth + h)},
   \end{equation}
   where $\xi_x$ is calculated similarly to \eqref{yhtaloryhma} by interpreting $x=\textit{o}_E$, \textcolor{black}{and the distance $d_{x}$ can be calculated accordingly.} For a general beam, the antenna angles can be calculated by first using a homography function, which maps the UEs into the nadir-pointing beam, and then using \eqref{eq:varphixxxx}.

   \section{LARGEST GAIN}

   \label{appendix:largestgain}
   Let $\px_{0} \in \mathcal{G}_1 \subset (0,1)$ be the largest gain. For convenience, assume a single access network ($p=1$). The PDF of $\px_0$ is determined by the void probability of the PPP:
   \begin{align}
     \label{eq:nthstrongestpdfY}
     \frac{d}{d x}\mathbb{P}(\px_{0} < x) &= \frac{d}{d x}  \exp\left\{-\int_{x}^{1}\lambda_{\mathcal{G}_1}(r)dr\right\}\nonumber \\
     &=\frac{d}{d x}\exp\left\{-\int_{x}^{1}\tilde{\kappa}/r dr\right\}  
       = \frac{d}{d x}x^{\tilde{\kappa}}= \tilde{\kappa}x^{\tilde{\kappa}-1}. 
   \end{align}
   The Laplace transform of $1/\px_0$ is
   \begin{equation}
     \mathcal{L}_{1/\px_0}(s) = \mathbb{E}(e^{-s/\px_0 }) = \tilde{\kappa}\int_0^1 x^{\tilde{\kappa}-1} e^{-s/x}dx = \tilde{\kappa} E_{\tilde{\kappa}+1}(s).
   \end{equation}

   \bibliographystyle{IEEEtran}
   \bibliography{IEEEabrv,source}

   \begin{IEEEbiography}[{\includegraphics[width=1in,height=1.25in,clip,keepaspectratio]{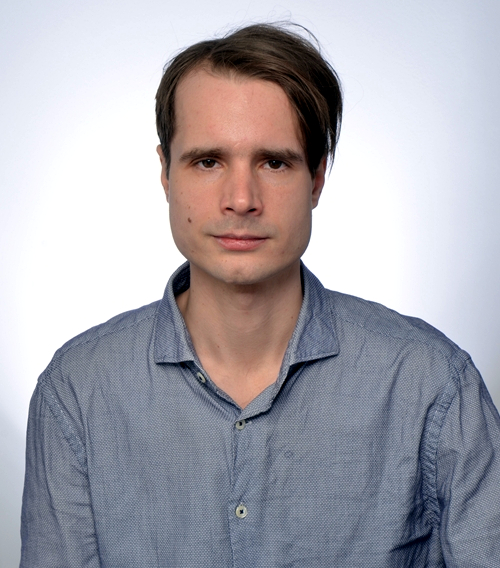}}]{Ilari Angervuori} Ilari received the B.Sc. and M.Sc. degrees in applied mathematics from the University of Helsinki, Finland, in 2016 and 2018, respectively.

     Since 2018, he has been with the School of Electrical Engineering, Aalto University, initially as a research assistant and then as a Ph.D. degree candidate under the supervision of Prof. Risto Wichman. His Ph.D. research subject is stochastic geometry in LEO networks, and he is one of the pioneering authors in stochastic geometry modeling of LEO, specializing primarily in the uplink. In academic research, he thrives on tractable system models and clear analytical insights.  From the summer of $2023$ to the winter of $2024$, he worked and published at the University of Notre Dame under the supervision of Prof. Martin Haenggi. He is an active peer reviewer in multiple IEEE journals.
 
     Mr. Angervuori holds an M.Sc. thesis approved with exceptional praise and is a member of the University of Helsinki Alumni Association.
   \end{IEEEbiography}

   \begin{IEEEbiography}[{\includegraphics[width=1in,height=1.25in,clip,keepaspectratio]{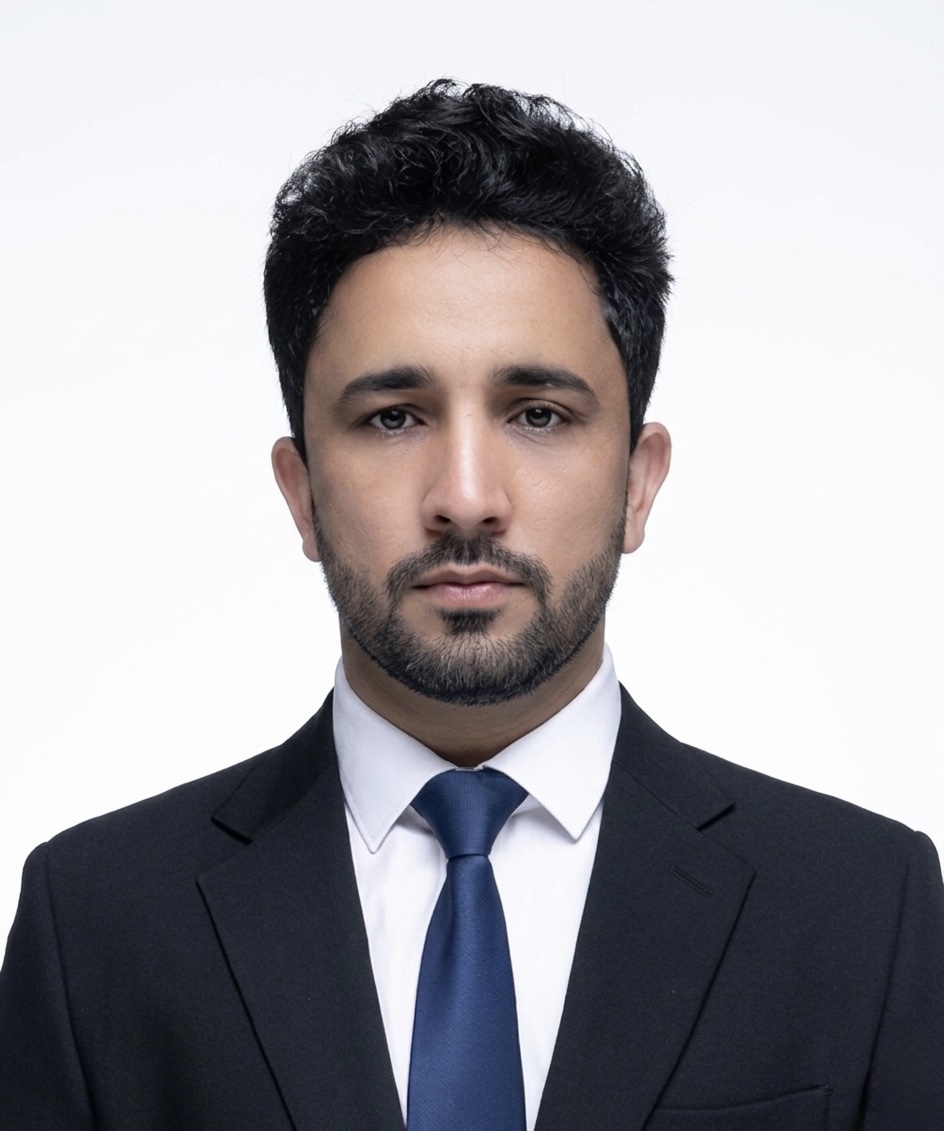}}]{Abid Afridi}

     Abid Afridi is currently a Marie Skłodowska-Curie Actions (MSCA) Doctoral Fellow with the Department of Information and Communications Engineering, Aalto University, Finland.

     He received the B.E. degree in Electrical Engineering (Telecommunication) from the University of Science and Technology Bannu, Pakistan, in 2020, and the Master’s degree in Electrical Engineering from the University of Ulsan, South Korea, in 2025, where he was a recipient of the Brain Korea 21 Plus (BK21+) Scholarship. From 2022 to 2025, he served as a Research Assistant with the Department of Electrical, Electronic and Computer Engineering, University of Ulsan. His research interests include sixth-generation (6G) non-terrestrial networks (NTNs), network slicing, resource management, wireless-powered Internet of Things (IoT) networks, and physical-layer security.

   \end{IEEEbiography}

   \begin{IEEEbiography}[{\includegraphics[width=1in,height=1.25in,clip,keepaspectratio]{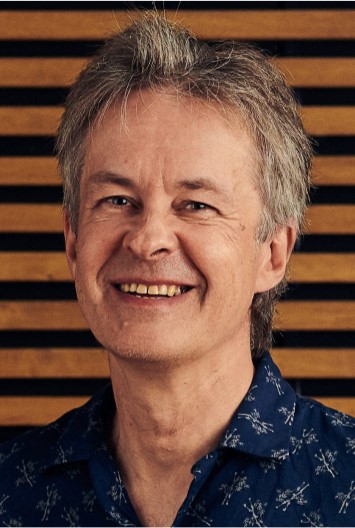}}]{Risto Wichman} received his M.Sc. and D.Sc. (Tech) degrees in electrical engineering from Tampere University of Technology, Finland, in 1990 and 1995, respectively.

     From 1995 to 2001, he worked at Nokia Research Center as a senior research engineer. In 2002, he joined the Department of Information and Communications Engineering, Aalto University School of Electrical Engineering, Finland, where he has been a full professor since 2008. His research interests lie in the physical layer of wireless communication systems, specifically in signal processing and communication theory.

   \end{IEEEbiography}

\EOD

\end{document}

%% file: testmathcommand.tex
\chardef\bslash=`\\ 

\newtheorem{thm}{Theorem}
\newtheorem{cor}[thm]{Corollary}
\newtheorem{lem}[thm]{Lemma}
\newtheorem{prop}[thm]{Proposition}

\newtheorem{example}{Example}

\theoremstyle{definition}

\theoremstyle{remark}
\newtheorem{rem}{Remark}

\newcommand{\eval}[2][\right]{\relax
  \ifx#1\right\relax \left.\fi#2#1\rvert}

